\documentclass[12pt]{article}

\usepackage{times} 
\usepackage{comment,amsfonts,amsmath,amsthm,graphicx}
\newcommand{\commentout}[1]{}

\usepackage{authblk}
\usepackage[ruled,vlined]{algorithm2e}
\usepackage{hyperref}       
\usepackage{fullpage}

\def\tO{\tilde{O}}
\def\epsilon{\varepsilon}
\def\eps{\varepsilon}

\newcommand{\alert}[1]{\textbf{\color{red}
		[[[#1]]]}\marginpar{\textbf{\color{red}**}}\typeout{ALERT:
		\the\inputlineno: #1}}

\def\eps{\epsilon}
\def\tO{\tilde{O}}

\newcommand{\R}{\mathbb{R}}

\newcommand{\E}{{\mathbb{E}}}

\newcommand{\mommit}[1]{}
\newcommand{\namedref}[2]{\hyperref[#2]{#1~\ref*{#2}}}

\newcommand{\theoremref}[1]{\namedref{Theorem}{#1}}

\newcommand{\figureref}[1]{\namedref{Figure}{#1}}

\theoremstyle{plain}
\newtheorem{theorem}{Theorem}
\newtheorem{lemma}{Lemma}

\newtheorem{remark}{Remark}
\newtheorem{claim}[lemma]{Claim}

\theoremstyle{definition}

\title{Almost Shortest Paths with Near-Additive Error in Weighted Graphs}

\begin{document}

\author[1]{Michael Elkin}
\author[1]{Yuval Gitlitz}
\author[1]{Ofer Neiman}

\affil[1]{Department of Computer Science, Ben-Gurion University of the Negev,
	Beer-Sheva, Israel. Email: \texttt{\{elkinm,neimano\}@cs.bgu.ac.il}, \texttt{gitlitz@post.bgu.ac.il }}

\maketitle
\begin{abstract}
Let $G=(V,E,w)$ be a weighted undirected graph with $n$ vertices and $m$ edges, and fix a set of $s$ sources $S\subseteq V$. We study the problem of computing {\em almost shortest paths} (ASP) for all pairs in $S \times V$ in both classical centralized and parallel (PRAM) models of computation. Consider the regime of multiplicative approximation of $1+\eps$, for an arbitrarily small constant $\eps > 0$ (henceforth $(1+\eps)$-ASP for $S \times V$). In this regime existing centralized algorithms require 
$\Omega(\min\{|E|s,n^\omega\})$ time, where $\omega < 2.372$ is the matrix multiplication exponent. Existing PRAM algorithms with polylogarithmic depth (aka time) require work $\Omega(\min\{|E|s,n^\omega\})$.

In a bold attempt to achieve centralized time close to the lower bound  of $m + n s$, Cohen \cite{C00} devised an algorithm which, in addition to the multiplicative stretch of $1+\eps$, allows also {\em additive error} of $\beta \cdot W_{\max}$, where 
$W_{\max}$ is the maximum edge weight in $G$ (assuming that the minimum edge weight is 1), and $\beta = (\log n)^{O({\frac{\log 1/\rho}{\rho}})}$ is polylogarithmic in $n$. It also depends on the (possibly) arbitrarily small parameter $\rho >0$ that determines the running time $O((m + ns)n^\rho)$ of the algorithm.

The tradeoff of \cite{C00} was improved in \cite{E01}, whose algorithm has similar approximation guarantee and running time, but its $\beta$ is $(1/\rho)^{O({\frac{\log 1/\rho}{\rho}})}$. However, the latter algorithm produces distance estimates rather than actual approximate shortest paths.
Also, the additive terms in \cite{C00,E01} depend linearly on a possibly quite large {\em global} maximum edge weight $W_{\max}$. 

In the current paper we significantly improve this state of affairs. Our centralized algorithm has running time $O((m+ ns)n^\rho)$, and its PRAM counterpart has polylogarithmic depth and work $O((m + ns)n^\rho)$, for an arbitrarily small constant $\rho > 0$. 
For  a pair $(s,v) \in S\times V$, it provides a path of length $\hat{d}(s,v)$ that satisfies 
$\hat{d}(s,v) \le (1+\eps)d_G(s,v) + \beta \cdot W(s,v)$, where $W(s,v)$ is the weight of the heaviest edge on some shortest $s-v$ path. Hence our additive term depends linearly on a {\em local} maximum edge weight, as opposed to the {\em global} maximum edge weight in \cite{C00,E01}. Finally, our $\beta = (1/\rho)^{O(1/\rho)}$, i.e., it is significantly smaller than in \cite{C00,E01}.

We also extend a centralized algorithm of Dor et al. \cite{DHZ00}. For a parameter $\kappa = 1,2,\ldots$, this algorithm provides for {\em unweighted} graphs a purely additive approximation of $2(\kappa -1)$ for {\em all pairs shortest paths} (APASP) in time $\tO(n^{2+1/\kappa})$.
Within the same running time, our algorithm for {\em weighted} graphs provides a purely additive error of $2(\kappa - 1) W(u,v)$, for every vertex pair $(u,v) \in {V \choose 2}$, with $W(u,v)$ defined as above.

On the way to these results we devise a suite of novel constructions of spanners, emulators and hopsets.
\end{abstract}
\thispagestyle{empty}
\newpage
\setcounter{page}{1}
\section{Introduction}


We study the problem of computing {\em almost shortest paths} from a set $S \subseteq V$ of designated vertices, called {\em sources}, to all other vertices in an $n$-vertex $m$-edge weighted undirected graph $G = (V,E,w)$, with non-negative edge weights. 
We aim at approximation guarantee of the form $(1+ \eps,\beta \cdot W)$, meaning that for every pair $(s,v) \in S \times V$, our algorithm will return a path of length $\hat{d}(s,v)$ (called {\em distance estimate}) that satisfies
\begin{equation}
\label{eq:near_add}
d_G(s,v) \le \hat{d}(s,v) \le (1+ \eps) d_G(s,v) + \beta \cdot W(s,v)~, 
\end{equation}
where $W(s,v)$ is the weight of the heaviest edge on some shortest $s-v$ path in $G$ (If there are multiple shortest paths, pick the one with minimal $W(s, v)$). 

Here $\eps > 0$ is an arbitrarily small positive constant, and $\beta$ is typically a large constant that depends on $\eps$, and on the running time of the algorithm.\footnote{In the introduction we will mostly suppress the dependence on $\eps$. It can however be found in the technical part of the paper.}
We call this problem {\em $(1 + \eps,\beta \cdot W)$-ASP for $S \times V$}. 

The problem of computing approximate shortest paths is one of the most central, fundamental and well-researched problems in Graph Algorithms. We study this problem in both the classical centralized and in the parallel (PRAM CRCW) models of computation.
Next, we overview the main previous results for this problem in these two settings, and describe our new results.
We start with the centralized model, and then turn our attention to the PRAM model.

\subsection{Centralized Setting}
\label{sec:centr}

The classical algorithm of Dijkstra for the exact {\em single-source shortest path} (SSSP) problem provides running time of $O(m + n \log n)$ \cite{FT87}. Thorup \cite{T97} devised an algorithm with running time $O(m + n\log\log n)$ for the case when all edge weights are integer. Using this algorithm separately from each of the $s = |S|$ sources results in running time of at least $(m \cdot s)$. 

On the opposite end of the spectrum, one can compute $(1+\eps)$-APASP (All Pair Approximate Shortest Paths) using matrix multiplication in $\tO(n^\omega)$ time \cite{GM97,AGM97,Z02}, where $\omega < 2.372$ is the matrix multiplication exponent \cite{CW90,W12,G14}. 
There are also a number of combinatorial (i.e., not exploiting fast matrix multiplication) APASP algorithms. 
In particular, Cohen and Zwick \cite{CZ01} showed that 3-APASP can be computed in $\tO(n^2)$ time. (They also provided a few additional algorithms with approximation ratio between 2 and 3 and running time greater than $n^2$.)
Baswana and Kavitha \cite{BK06} improved their approximation guarantee to $(2,W)$ (with the same running time of $\tO(n^2)$),
and
with $W$ defined as in (\ref{eq:near_add}). 

Finally, Cohen \cite{C00} devised an algorithm that for an arbitrarily small parameter $\rho > 0$, solves $(1+\eps,\beta \cdot W_{\max})$-ASP for $S \times V$ in $\tO((m + s\cdot n) \cdot n^\rho)$ time\footnote{The running time of the algorithm of \cite{C00}, as well as of the algorithm of \cite{E01} and of our algorithm, is actually slightly better than $O((m+sn)n^\rho)$. Specifically, it is roughly $O(m n^\rho + s n^{1+1/2^{1/\rho}})$. We use the simpler expression of $O((m + sn)n^\rho)$ to simplify presentation. More precise and general bounds can be found in the technical part of the paper.}, with 
$\beta = \beta_{Coh} = (\log n)^{O({{\log (1/\rho)} \over \rho})}$. Here $W_{\max}$ is the maximum edge weight in the entire graph (assuming the minimum edge weight is 1). 

This result was improved by Elkin \cite{E01}. The algorithm of Elkin \cite{E01} also solves $(1+\eps,\beta \cdot W_{\max})$-ASP for $S \times V$ in $\tO((m + sn)n^\rho)$ time, with $\beta = \beta_{Elk} = (1/\rho)^{O(\frac{\log(1/\rho)}{\rho})}$. 
This algorithm 
 reports distance estimates, rather than actual paths.\footnote{There is a variant of the algorithm of \cite{E01} which reports actual paths, but requires time $\tO((m + sn)n^\rho \cdot W_{\max})$. This running time is  typically prohibitively large as it depends linearly on $W_{\max}$.} 

The running time of \cite{C00,E01} is close (up to $n^\rho$, for an arbitrarily small constant $\rho >0$) to the lower bound of $\Omega(m + n\cdot s)$. This is unlike other aforementioned algorithms, whose running time is much larger, i.e., 
$\Omega(\min\{m \cdot s,n^2\})$. However, the algorithms of \cite{C00,E01} suffer from a number of drawbacks.
First, their additive term is linear in the maximum edge weight  $W_{\max}$. Second, the coefficient $\beta$ in them is quite large, even for a relatively large values of the parameter $\rho$. Third, as was mentioned above, 
the algorithm of \cite{E01}
  returns distance estimates, as opposed to actual paths that implement these estimates, and in the algorithm of \cite{C00} the coefficient $\beta$ of the additive term is super-constant (specifically, polylogarithmic in $n$). 

In the current paper we address these issues, and devise an algorithm for \((1 + \eps,\beta\cdot W)\)-ASP for \(S \times V\) with running time \(O((m +s n)n^\rho)\), for an arbitrarily small parameter \(\rho > 0\), with \(\beta = (1/\rho)^{O(1/\rho)}\).
This algorithm does report paths, rather than just distance estimates. Note also that the additive term grows linearly with the {\em local} maximum edge weight, i.e., with the weight of heaviest edge on each particular source-destination shortest path, as opposed to the {\em global}  maximum edge weight $W_{\max}$.
Finally, its coefficient $\beta$ is significantly smaller than $\beta_{Elk} = (1/\rho)^{O({{\log (1/\rho)}  \over \rho})}$, though it is admittedly still quite large\footnote{We are able to further decrease $\beta$ to $2^{O(1/\rho)}$, at the expense of increasing the multiplicative stretch from $1+\eps$ to $3+\eps$.}. (The coefficient $\beta_{Coh}$ of \cite{C00} depends polylogarithmically on $n$, while the coefficient $\beta$ in \cite{E01} and here are independent of $n$.)  

We also extend an algorithm of Dor et al. \cite{DHZ00} to weighted graphs. Specifically, the algorithm of \cite{DHZ00} works for {\em unweighted} undirected graphs. For any parameter $\kappa = 1,2,\ldots$, it provides an additive $2(\kappa - 1)$-APASP in $\tO(n^{2+1/\kappa})$ time. Our extension applies to {\em weighted} undirected graphs.
It computes additive $2(\kappa-1)W$-approximation for all pairs shortest paths within the same time $\tO(n^{2+1/\kappa})$, i.e., for any vertex pair $u,v \in V$, it produces a path of length at most $d_G(u,v) + 2(\kappa - 1)\cdot W(u,v)$, where $W(u,v)$ is as in (\ref{eq:near_add}). 

Note that the linear dependence of additive error on $W$ is unavoidable, as an algorithm with stretch $(1+\eps,o(W))$ can be translated into an algorithm with the same running time and with a purely multiplicative stretch of $1+\eps$.

\subsection{Parallel Setting}
In the PRAM model, multiple processors are connected to a single memory block, and the operations
are performed in parallel by these processors. We will mostly be concerned with the Concurrent Read Concurrent Write (CRCW) PRAM model, 
that allows multiple processors to access any memory cell at any given round. The running time is measured by the number of rounds, and the work by
the number of processors multiplied by the number of rounds.

Early algorithms for these problems \cite{UY91,KS93,SS99,S97} require $\Omega(\sqrt{n})$ parallel time.
Algorithms of \cite{GM97,AGM97,Z02}, that were discussed in Section \ref{sec:centr},
can also be  applied in PRAM. They provide  $(1+\eps)$-APASP in polylogarithmic time and $\tO(n^\omega)$ work. The algorithm 
of Cohen \cite{C00}, for a parameter $\rho > 0$, solves $(1+\eps)$-ASP for $S \times V$ in polylogarithmic time 
$(\log n)^{O({{\log (1/\rho)} \over \rho})}$ and work $\tO((m + n^{1+\rho}) s + m \cdot n^\rho)$.

This tradeoff was then improved in \cite{EN16,EN19hop}, where the running time is 
$(\log n)^{O(1/\rho)}$, and the work is the same as in \cite{C00}. Further spectacular progress was recently achieved 
by \cite{L20,ASZ20}, who devised $(1+\eps)$-SSSP algorithms with time $(\log n)^{O(1)}$ and work $\tO(m)$. Nevertheless, for $(1+\eps)$-ASP problem from the set $S \subseteq V$ of sources, one needs to run these algorithms in parallel from all the $s$ sources. As a result, their work complexity becomes $\tilde{\Theta}(m s)$. 

To summarize, all existing solutions for the problem with polylogarithmic time have work complexity $\Omega(\min\{m \cdot s,n^\omega\})$. We devise the first algorithm with polylogarithmic time $(\log n)^{O(1/\rho)}$ and work complexity 
$\tO((m + n s) n^\rho)$, for an arbitrarily small constant $\rho > 0$. In other words,  our work complexity is within $n^\rho$, for an arbitrarily small constant $\rho>0$, close to the lower bound of $\Omega(m + ns)$. 
On the other hand, unlike algorithms of \cite{Z02,C00,EN16,EN19hop,L20,ASZ20}, whose approximation guarantee is a purely multiplicative $1 + \eps$, for an arbitrarily small $\eps > 0$, our approximation guarantee is $(1+\eps,\beta\cdot W)$, with
$\beta = (1/\rho)^{O(1/\rho)}$, and $W$ as in (\ref{eq:near_add}). 

Moreover, our result can, in fact, be viewed as a {\em PRAM distance oracle}. Specifically, following the preprocessing that requires 
time $(\log n)^{O(1/\rho)}$ and work $\tO(m n^\rho)$, our algorithm stores a compact data structure of size $\tO(n^{1+\rho})$.
Given a query  vertex $s$, this data structure provides distance estimates $\hat{d}(s,v)$ for all $v \in V$, which satisfies (\ref{eq:near_add}) in {\em constant time} and using work $\tO(n^{1+\rho})$, where $\beta = (1/\rho)^{O(1/\rho)}$.
Note that the distance oracle has size arbitrarily close to linear in $n$, its preprocessing time is polylogarithmic and preprocessing work is arbitrarily  close to linear in $m$, and the query time is constant and the query work is arbitrarily close to linear in $n$.
(Note that as the query provides distance estimates for $n$ vertex pairs $\{(s,v) \mid v \in V\}$, its query work complexity must be $\Omega(n)$.) 


\subsection{Hopsets, Spanners and Emulators}

From the technical viewpoint, these results are achieved via a combination of our new constructions of emulators, spanners and hopsets. 
For parameters $\alpha,\beta\ge 1$, we say that a graph $H=(V,E', w)$ is an {\em $(\alpha,\beta)$-hopset} for a (weighted) graph $G=(V,E, w)$, if by adding $E'$ to the graph, every pair $x,y\in V$ has an $\alpha$-approximate shortest path consisting of at most $\beta$ hops; Formally,
\[
d_G(x,y)\le d_{G\cup H}^{(\beta)}(x,y)\le \alpha\cdot d_G(x,y)~,
\]
where $d_{G\cup H}^{(\beta)}$ is the shortest path in $G\cup H$ containing at most $\beta$ edges. The parameter $\alpha$ is called the {\em stretch}, and $\beta$ is the {\em hopbound}. 

We say that $H$ is an $(\alpha,\beta)$-{\em emulator} if for every $x,y\in V$,
\[
d_G(x,y)\le d_H(x,y)\le \alpha\cdot d_G(x,y)+\beta~,
\]
and $H$ is a {\em spanner} if it is an emulator and a subgraph of $G$.

Hopsets and near-additive spanners are fundamental combinatorial constructs, and play a major role in efficient approximation of shortest paths in various computational models. These objects have been extensively investigated in recent years \cite{EP04,E04,TZ06,EZ06,P09,P10,BKMP10,EN19span,KS97,C00,B09,N14,HKN14,MPVX15,HKN16,FL16,EN16,ABP17,EN19hop,HP19}. The main interest is to understand the triple tradeoff between the size of the hopset (respectively, spanner), to the stretch $\alpha$, and to the hopbound (resp., additive stretch) $\beta$. For algorithmic applications, it is also crucial to bound the construction time of the hopset/spanner/emulator.


We show near-additive spanners for {\em weighted} graphs, where the additive stretch for the pair $x,y$ may depend also on the largest edge weight on the corresponding shortest path from $x$ to $y$,  $W(x,y)$. For a parameter $0<\rho<1$, we devise an algorithm that constructs a $(1+\epsilon,\beta \cdot W)$-spanner of size $O(n^{1+1/2^{1/\rho}}+ n/\rho)$ with $\beta\le  \left(\frac{1/\rho}{\epsilon}\right)^{O(1/\rho)}$. We also show how to analyze the construction so that it yields smaller additive stretch, while increasing the multiplicative one. Specifically, we get a
$(c,\beta\cdot W)$-spanner of same size as above, for every constant $c>3$, and with $\beta= \left(\frac{1}{\epsilon}\right)^{O(1/\rho)}$. 
Our emulators have a somewhat improved $\beta$. 
All of our results admit near-linear time algorithms, i.e., their running time is $O(|E|n^\rho)$, for an arbitrarily small constant $\rho>0$.


\subsection{Technical Overview}

We adapt the constructions of \cite{TZ06,P09,EN19span,EN19hop,HP19} of hopsets, spanners and emulators so that they are suitable for weighted graphs, and provide an improved additive stretch (or hopbound) $\beta$. The basic idea in all these constructions is to generate a random hierarchy of vertex sets $V=A_0\supseteq A_1\supseteq\dots\supseteq A_k=\emptyset$, where for each $0\le i<k-1$, each element in $A_i$ is included in $A_{i+1}$ with probability $\approx n^{-2^{i-k}}$ (one should think of $k=1/\rho$). We also refer to vertices at $A_i$ as vertices at level $i$. For each $v\in V$, the pivot $p_i(v)$ is the closest vertex in $A_i$ to $v$. Then the set of edges $H$ is created by connecting, for every $0\le i\le k-1$, every vertex $v\in A_i\setminus A_{i+1}$ to its {\em bunch}: all other vertices in $A_i$ that are closer to it than $p_{i+1}(v)$. One difference in our construction is that we also connect each vertex to all its (at most $k$) pivots. The main technical innovation in this work is the new analysis of this construction, yielding various hopsets, emulators and spanners that apply for {\em weighted} graphs, and admit improved parameters. (Previous constructions of spanners and emulators \cite{EP04,TZ06,P09,EN19span, HP19} applied only for unweighted graphs.)

The previous analysis of the stretch for some pair $x,y\in V$ goes roughly as follows. Divide the $x-y$ path into intervals, and try to connect these intervals using low stretch paths in $H$ (and for hopsets, also some of the graph edges, but with few hops). Each interval can either have a low stretch path; or fail, in which case that interval admits a nearby pivot (of some level $i$). Then, consider two failed intervals (the leftmost and rightmost ones), and try to find an $x-y$ path via the pivots of these intervals.



We note that this partitioning of the $x-y$ path to equal size intervals (used in previous works), cannot work directly for weighted graphs, since it may not be possible to divide the path to  intervals of equal (or even near-equal) size. We provide a subtle adaptation of this technique so that it can handle weights. In particular, we distinguish between short and long distances: The sufficiently long distances may suffer a partition to un-equal intervals, as the induced error is dominated by the multiplicative stretch of  $1+\epsilon$. For the short distances, we stop this partitioning when it becomes too ''expensive", and resort to an argument similar to the one in \cite{TZ01}, which has large multiplicative stretch (of roughly $2^{1/\rho}$). However, at the point where we stop, that stretch can be accounted for by the additive stretch $\beta\cdot W$.

An additional ingredient in our new analysis of $H$ as an emulator/spanner for weighted graphs, given a pair $x,y\in V$, is to iteratively find a vertex $z$ on the $x-y$ shortest path (sufficiently far from $x$) that admits in $H$ a path with low multiplicative stretch from $x$. When there is no such $z$, we show that we can reach $y$ with small additive stretch. This technique can be used to obtain the improved dependence of $\beta$ on the parameter $k=1/\rho$.

Recall that emulators are insufficient for reporting paths. In particular, the approach of \cite{C00,E01} was based on emulators, rather than on spanners, and it is not clear if these algorithms can be adapted to build spanners (for weighted graphs).
On the other hand, with our approach we can convert our constructions of emulators into constructions of spanners.  
Specifically, to build a spanner, we  must use  graph edges. So in order to connect vertices $v\in A_i\setminus A_{i+1}$ to the vertices in their bunch $B(v)$, we need to add paths of possibly many edges in the graph (rather than a single edge, as for emulators/hopsets). 
To do this we connect every vertex $v$ to all vertices in its {\em half-bunch} (see Section \ref{sec:spanner} for its definition), as opposed to connecting it to all vertices in its full bunch. (The latter is the case in the construction of emulators.) This turns out to be sufficient to ensure that the union of all these paths does not contain too many edges. For this analysis we employ ideas of counting pairwise path intersections, developed in the context of distance preservers \cite{CE05} and near-additive spanners and distance oracles for unweighted graphs \cite{P09,EP15}. We simplify this approach, and extend it to weighted graphs.


\subsection{Organization}

In Section \ref{sec:const} we describe the construction of our emulators and hopsets. In Section \ref{sec:emulator} we analyze this construction, showing it provides emulators for weighted graphs.
In Section \ref{sec:spanner} we describe the construction of our spanners for weighted graphs (proofs are deferred to Appendix \ref{appendix:spanner_proofs}) and use them for our centralized algorithms for the ASP problem.
In Section \ref{sec:apps} we provide  efficient implementations of our constructions, and use them in Section \ref{sec:ASP} for solving ASP in PRAM, and for PRAM distance oracles.

Our hopsets and emulators with improved $\beta$ appear in Appendices \ref{sec:hopset} and \ref{sec:3+eps-emul}.
In Appendix \ref{sec:dhz} we devise an algorithm for pure additive APASP for weighted graphs. 


\subsubsection{Bibliographical note}
Related results about $(\alpha,\beta)$-spanners for {\em unweighted} graphs and $(\alpha,\beta)$-hopsets with $\alpha \ge 3+\eps$ were achieved independently of us and simultaneously by \cite{BLP20}. 
In another submission \cite{EN20} 
, an $(1+\eps)$-ASP algorithm for $S \times V$ with $|S| = n^r$,
for some $0 < r \le 1$. was devised. The running time of this algorithm in the centralized setting is $\tO(n^{\omega(r)})$, where
$\omega(r)$ is the rectangular matrix multiplication exponent. ($n^{\omega(r)}$ is the time required to multiply an $n^r \times n$ matrix by an $n \times n$ one.) In PRAM setting that algorithm runs in polylogarithmic time and  $\tO(n^{\omega(r)})$ work.
For graphs $G = (V,E,w)$ and sets of sources  $S \subseteq V$ of size $s$ such that $m + ns = o(n^2)$, the algorithms that we devise in the current submission are more efficient than in \cite{EN20}.

Following our work, \cite{ABDKS20} devised algorithms for purely additive spanners (e.g., spanners with multiplicative stretch exactly 1) for weighted graphs. In their results, the additive term depends on the global maximal edge weight $W_{\max}$, and the size of their spanners is always $\Omega(n^{4/3})$ (the latter is unavoidable for purely additive spanners, due to a lower bound of \cite{AB16}).


\section{Construction}\label{sec:const}
We use a similar construction to that of \cite{TZ06,EN19hop,HP19}. One difference is that every vertex connects to pivots in all levels. (Recall that the main difference is in the analysis.) 
Let $G = (V, E, w)$ be a weighted graph with $n$ vertices, and fix an integer parameter $k \geq 1$.
Let $\nu = 1/(2^k-1)$.
Let $A_0\dots A_k$ be sets of vertices such that $A_0 = V$, $A_k = \emptyset$, and for $0 \leq i \leq k-2$, $A_{i+1}$ is created by sampling independently every element from $A_i$ with probability $q_i = n^{-2^i\nu}\cdot 2^{-2^i -1}$.
For every $0 \leq i \leq k-1$, the expected size of $A_i$ is:

\[ N_i = E[|A_i|] = n \prod_{j=0}^{i-1}q_j = n^{1-(2^i-1)\nu} \cdot 2^{-2^{i} - i + 1}\]

For every $i \in [k-1]$, define the pivot $p_i(v)$ to be the closest vertex in $A_i$ to $v$, breaking ties by lexicographic order. For every $u \in A_i \backslash A_{i+1}$ define the bunch (see \figureref{fig:bunch})
\begin{equation}\label{eq:bunch}
B(u) = \{v \in A_i | d_G(u, v) < d_G(u, A_{i+1})\} \cup \{p_j(u) | i < j < k \}~.
\end{equation}
That is, the bunch $B(u)$ contains all the vertices which are in $A_i$ and closer to $u$ than $p_{i+1}(u)$, and at most $k$ pivots.
We then define $H = \{(u, v) ~:~ u \in V, v \in B(u)\}$, where the weight of the edge $(u, v)$ is set as the weight of the shortest path between $u,v$ in $G$.

\begin{figure}[ht]
	\centering
	\includegraphics[scale=0.3]{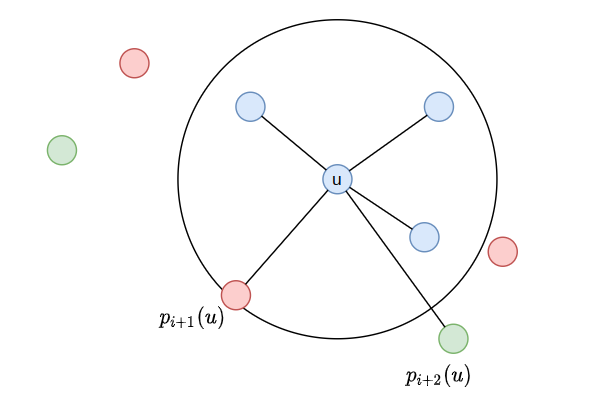}	
	\caption{The bunch of $u$. Here vertices in $A_i \backslash A_{i+1}$ are colored in blue, in $A_{i+1} \backslash A_{i+2}$ are colored in red and in $A_{i+2}$ are colored in green.}
	\label{fig:bunch}
\end{figure}

\begin{lemma}
	\label{lem:size_hopset_emu}
	The size of $H$ is 
	\begin{align*}
	|H| = O(kn + n^{1+\nu})
	\end{align*}
\end{lemma}

The proof of Lemma \ref{lem:size_hopset_emu} is similar to previous works, we include it for completeness in appendix \ref{size_analysis_hop_emu}.

\section{Near-Additive Emulators for Weighted Graph}\label{sec:emulator}


In this section we will show that $H$ as defined in Section \ref{sec:const} is a $(1 + \varepsilon,\beta\cdot W)$-emulator for weighted graphs. Note that the construction there does not depend on $\varepsilon$, and indeed, $H$ will be an emulator for all values of $0<\varepsilon <1$ simultaneously (with $\beta=\beta(\varepsilon)$ depending on $\varepsilon$). 


Let $G = (V, E)$ be a weighted graph with non-negative weights $w:E\to\R_+$, recall that for $x,y\in V$ we have $W(x,y) = \max\{w(e) ~:~ e \in P_{xy}\}$ (where $P_{xy}$ is a shortest path from $x$ to $y$ in $G$). Let $k \geq 1$ and $\Delta > 3$ be given parameters (think of $\Delta=3+O(k/\epsilon)$).
We begin by proving two lemmas, handling long and short distances, respectively. The first lemma asserts that pairs which are sufficiently far apart admit either a low stretch path, or a nearby pivot of a higher level.
\begin{lemma}\label{1epsemu_constant_lemma}
	Fix $\Delta >3$. Let $0 \leq i < k$ and let $x,y \in V$ such that $d_G(x,y) \geq (3 \Delta) ^ iW(x,y)$. Then at least one of the following holds:
	\begin{enumerate}
		\item $d_H(x, y) \leq (1 + \frac{4i}{\Delta - 3}) d_G(x, y)$
		\item $d_H(x, p_{i+1}(x)) \leq \frac{\Delta}{\Delta - 3} d_G(x, y)$
	\end{enumerate}
	(For $i=k-1$, the first item must hold since $p_{i+1}$ doesn’t exist).
\end{lemma}
\begin{proof}
	Denote $W=W(x,y)$.
	The proof is by induction on $i$. For the base case $i=0$, if $y \in B(x)$ then $d_H(x,y)=d_G(x,y)$ and the first item holds. Otherwise, if $x\in A_1$ then $d_G(x,p_1(x))=0$, so the second item holds. The last case is that $x\in A_0\setminus A_1$ and $y\notin B(x)$, then by (\ref{eq:bunch}) we have $d_H(x, p_1(x)) =d_G(x, p_1(x)) \leq d_G(x, y)\le\frac{\Delta}{\Delta - 3} d_G(x, y) $, thus the second item holds. Assume the lemma holds for $i$ and prove for $i+1$. Let $x,y\in V$ be a pair of vertices such that  $d_G(x,y) \geq (3 \Delta) ^ {i+1}W(x,y)$
	
	Divide the shortest path between $x$ and $y$ into $J$ segments $\{L_j = [u_{j}, u_{j+1}]\}_{j \in [J]}$ of length at least $(3\Delta)^iW$ and at most $d_G(x,y)/\Delta$. It can be done as follows: define $u_1=x, j=2$ and walk on the shortest path from $x$ to $y$. Define $u_j$ as the first vertex which \(d_G(u_{j-1}, u_j) \geq (3\Delta)^iW\) or define $u_j = y$ if $d_G(u_{j-1}, y) < (3\Delta)^iW$. Increase $j$ by 1 and repeat. Note that each segment has length at most $(3\Delta)^iW+W$. Finally, we join the last two segments.
	The length of the last segment is at most $(3\Delta) ^ iW + W + (3\Delta)^iW \leq 3^{i+1}\Delta^iW\leq d_G(x,y)/\Delta$.
	The length of any segment except the last is also at most $(3\Delta) ^iW + W \le d_G(x,y)/\Delta$.
	
	Apply the induction hypothesis for every segment $L_j$ with parameter $i$. If for all the segments the first item holds, then first item holds for $x, y$ and $i+1$, since
	\begin{align*}
	d_H(x,y) &\leq \sum_{j\in J} d_H(u_j, u_{j+1}) 
	\leq \sum_{j \in J}(1 + \frac{4i}{\Delta - 3})d_G(u_j, u_{j+1}) \\
	&\leq (1 + \frac{4i}{\Delta - 3}) d_G(x,y).
	\end{align*}
	
	Otherwise, for at least one segment the second item holds. Let $L_l$ (resp., $L_{r-1}$) be the leftmost (resp., rightmost) segment for which the second item holds. 
	We have that the second item holds for the pair $u_l,u_{l+1}$, and by symmetry of the first item, the second item also holds for the pair $u_{r}, u_{r-1}$ with parameter $i$. Hence
	
	\begin{align}
	d_H(u_r, p_{i+1}(u_r)) \leq \frac{\Delta}{\Delta - 3} d_G(u_{r-1}, u_r) \leq \frac{d_G(x, y)}{\Delta - 3} , \nonumber \\
	d_H(u_l, p_{i+1}(u_l)) \leq \frac{\Delta}{\Delta - 3} d_G(u_{l}, u_{l+1})\leq \frac{ d_G(x, y)}{\Delta - 3}.
	\label{pi_1epsemul}
	\end{align}
	
	Consider first the case that $p_{i+1}(u_r) \in B(p_{i+1}(u_l))$. In this case $H$ contains the edge $\{p_{i+1}(u_l), p_{i+1}(u_r)\}$, so we have
	\begin{align}
	d_H&(p_{i+1}(u_l), p_{i+1}(u_r)) 
	\leq d_G(p_{i+1}(u_l), u_l) + d_G(u_l, u_r) + d_G(u_r, p_{i+1}(u_r)).
	\label{between_pi_1epsemu}
	\end{align}
	
	By the triangle inequality,
	\begin{align}
	d_H&(u_l, u_r) \\
	& \leq d_H(u_l, p_{i+1}(u_l)) + d_H(p_{i+1}(u_l), p_{i+1}(u_r)) + d_H(p_{i+1}(u_r), u_r) \nonumber\\
	& \stackrel{(\ref{between_pi_1epsemu})}{\leq} 2d_H (u_l, p_{i+1}(u_l)) + d_G(u_l, u_r) + 2d_H(p_{i+1}(u_r), u_r) \nonumber\\
	& \stackrel{(\ref{pi_1epsemul})}{\leq} \frac{4 d_G(x, y)}{\Delta - 3} + d_G(u_l, u_r).
	\label{1epsemu_lr}
	\end{align}
	
	Thus, the distance between $x$ and $y$ in $H$,
	\begin{align*}
	d_H(x, y) 
	& \leq \sum_{j \in [J]} d_H(u_j, u_{j+1}) \\
	&\leq \sum_{j=1}^{l-1} \left(1 + \frac{4 i}{\Delta - 3}\right)d_G(u_j, u_{j+1}) + d_H (u_l, u_r) 
	+ \sum_{j=r}^{J}\left(1 + \frac{4 i}{\Delta - 3}\right)d_G(u_j, u_{j+1}) \\
	&\leq \left(1 + \frac{4 i}{\Delta - 3}\right)d_G(x, u_l) + d_H (u_l, u_r) + \left(1 + \frac{4 i}{\Delta - 3}\right)d_G(u_r, y) \\
	&\stackrel{(\ref{1epsemu_lr})}{\leq} \left(1 + \frac{4 i}{\Delta - 3}\right)d_G(x, u_l) + \frac{4d_G(x, y)}{\Delta - 3}  + d_G (u_l, u_r) 
	+ \left(1 + \frac{4 i}{\Delta - 3}\right)d_G(u_r, y) \\
	& \leq \left(1 + \frac{4(i+1)}{\Delta - 3}\right)d_G(x, y),
	\end{align*}
	therefore the first item holds.
	
	The other case is that $p_{i+1}(u_r) \notin B(p_{i+1}(u_l))$. Then
	\begin{equation}\label{eq:no-pivot-2}
	d_H(p_{i+1}(u_l), p_{i+2}(p_{i+1}(u_l))) \leq d_G(p_{i+1}(u_l), p_{i+1}(u_r))~.
	\end{equation}
	We can bound the distance $d_H(x, p_{i+2}(x))$ by
	\begin{align*}
	d_H&(x, p_{i+2}(x)) \\
	& \leq d_G(x, u_l) + d_G(u_l, p_{i+1}(u_l)) + d_G(p_{i+1}(u_l), p_{i+2}(p_{i+1}(u_l))) \\
	& \leq d_G(x, u_l) + d_G(u_l, p_{i+1}(u_l)) + d_G(p_{i+1}(u_l), p_{i+1}(u_r)) \\
	& \stackrel{(\ref{between_pi_1epsemu})}{\leq } d_G(x, u_l) + d_G(u_l, p_{i+1}(u_l)) 
	+ d_G(u_l, p_{i+1}(u_l)) + d_G(u_l, u_r) + d_G(p_{i+1}(u_r), u_r) \\
	& \stackrel{(\ref{pi_1epsemul})}{\leq} d_G(x, y) + \frac{3 d_G(x,y)}{\Delta - 3} = \frac{\Delta}{\Delta - 3} d_G(x,y).
	\end{align*}
	Hence the second item holds.
\end{proof}
The previous lemma is useful for vertices which are very far from each other, since for $i = k-1$ the first item must hold. For vertices which are close to each other, we will need the following lemma.

\begin{lemma}\label{1epsemu_beta_lemma}
	Let $0 \leq i < k$ and fix $x,y \in V$. Let 
	
	$ m =\max\{d_G(x, p_i(x)),d_G(y, p_i(y)),d_G(x,y)\}$. Then at least one of the following holds:
	\begin{enumerate}
		\item $d_H(x,y) \leq 5m$
		\item $d_H(x, p_{i+1}(x)) \leq 4m$ and $d_H(y, p_{i+1}(y)) \leq 4m$
	\end{enumerate}
	(For $i=k-1$, the first item must hold since $p_{i+1}$ doesn’t exist).

\end{lemma}
\begin{proof}
	If $p_i(y) \in B(p_i(x))$ the first item holds, since
	\begin{align*}
	d_H(x, y) &\leq d_H(x, p_i(x)) + d_H(p_i(x), p_i(y)) + d_H(p_i(y), y) \\
	& \leq d_G(x, p_i(x)) + d_G(p_i(x), x) + d_G(x, y) 
	+ d_G(y, p_i(y)) + d_G(p_i(y), y) \\
	&\leq 5m~.
	\end{align*}
	
	If $p_i(y) \notin B(p_i(x))$, then $d_G(p_i(x), p_{i+1}(p_i(x))) \leq d_G(p_i(x), p_i(y))$, in this case the second item holds, as
	\begin{align*}
	d_H(x, p_{i+1}(x)) 
	&\leq d_H(x, p_i(x)) + d_H(p_i(x), p_{i+1}(p_i(x))) \\
	& \leq d_G(x, p_i(x)) + d_G(p_i(x), x) + d_G(x, y) + d_G(y, p_i(y)) \leq 4m~.
	\end{align*}
	
	The bound on $d_H(y, p_{i+1}(y))$ is symmetric.
\end{proof}

We are now ready to prove the following theorem.

\begin{theorem}\label{1epsemu_theorem}
	For any weighted graph $G = (V, E)$ on $n$ vertices, and any integer $k > 1$, there exists $H$ of size at most $O(kn + n^{1+1/(2^k - 1)})$, which is a $(1+\varepsilon, \beta\cdot W)$-emulator for any $0<\varepsilon <1$, where $\beta = O( \frac{k}{\varepsilon})^{k-1}$.
\end{theorem}

\begin{proof}
	Fix $\Delta = 3 + \frac{4(k-1)}{\varepsilon}, \beta = 10 (3 \Delta)^{k-1}$. Let $x, y \in V$, and $W=W(x,y)$.
	If $d_G(x, y) \geq (3\Delta)^{k-1}W$, we can apply Lemma \ref{1epsemu_constant_lemma} for $x, y$ and $i=k-1$. Since $p_{k}(x)$ does not exist, the first item must hold.
	Thus,
	\[
	d_H(x,y) \leq \left(1 + \frac{4(k-1)}{\Delta - 3}\right) d_G(x, y) = (1 + \varepsilon)d_G(x,y)~.
	\]

	Otherwise, take the integer $0\le i<k-1$ satisfying $(3\Delta)^iW\le d_G(x, y) < (3\Delta)^{i+1}W$ (note that there must be such an $i$, since $d_G(x,y)\ge W$). Apply Lemma \ref{1epsemu_constant_lemma} for $x, y$ and $i$. If the first item holds, we will get $1+\varepsilon$ stretch as before.
	
	Otherwise, the second item holds, and we know that $d_G(x, p_{i+1}(x)) \le\frac{\Delta}{\Delta-3}d_G(x,y)\le 2d_G(x,y)$ (using that $k\ge 2$). By symmetry of $x,y$ in the first item of Lemma \ref{1epsemu_constant_lemma}, we have $d_G(y,p_{i+1}(y))\le 2d_G(x,y)$ as well.
	Set $j=i+1$ and apply Lemma $\ref{1epsemu_beta_lemma}$ with $x, y, j$, noting that $m\le 2d_G(x,y)$. If the first item holds, we found a path in $H$ from $x$ to $y$ of length at most $5m\le 10 d_G(x,y)\le 10 (3 \Delta)^{k-1}W = \beta \cdot W$.
	
	If the second item holds, we increase $j$ by one and apply Lemma $\ref{1epsemu_beta_lemma}$ again. We continue this procedure until the first item holds. The fact that the second item holds implies that the bound $m$ (the maximal distance of $x,y$ to the level $j$ pivots) increases every iteration by a factor of at most 4.
	Since the first item must hold for $j=k-1$, the path we found is of length at most $10\cdot 4^{k-i-2} d_G(x,y) \le 10\cdot 4^{k-i-2}\cdot (3\Delta)^{i+1}W$, which is maximized for $i=k-2$. Hence the additive stretch is at most $10\cdot(3\Delta)^{k-1}W = O((9 + \frac{12(k-1)}{\varepsilon})^{k-1} W)$, as required.
\end{proof}

\begin{remark}
	We note that the analysis did not use the fact that the path from $x$ to $y$ is a shortest path. In particular, for every path $P$ from $x$ to $y$ of length $d$, we can obtain a path in $H$ of length at most $(1+\varepsilon)\cdot d+\beta\cdot W(P)$, where $W(P)$ is the largest edge weight in $P$.
\end{remark}

\section{Near-Additive Spanners for Weighted Graphs}\label{sec:spanner}

In this section we devise our spanners for weighted graphs. We first describe the new construction, that differs from that of Section \ref{sec:const} in several aspects, which are required in order to keep the size of the spanner under control (and independent of $W_{\max}$). In particular, since we add paths, rather than edges, between each vertex to other vertices in its bunch, we need to ensure the size of $H$ is small enough. To do that, we use half-bunches rather than bunches to define $H$ (see (\ref{eq:hb}) below), and show that there are few intersections between the aforementioned paths (see Lemma \ref{lem:paths}). One last ingredient is altering the sampling probabilities, so that the argument on intersection goes through. This approach refines and improves ideas from \cite{P09,EP15}.

\paragraph{Construction.}
Let $G = (V, E)$ be a weighted graph with $n$ vertices, and fix an integer parameter $k \geq 3$.
Let $\nu = \frac{1}{(4/3)^k-1}$.
Let $A_0\dots A_k$ be sets of vertices such that $A_0 = V$, $A_k = \emptyset$, and for $0 \leq i \leq k-2$, $A_{i+1}$ is created by sampling every element from $A_i$ with probability $q_i = n^{-4^i\nu/3^{i+1}}$.
For every $0 \leq i \leq k-1$, the expected size of $A_i$ is:

\[ N_i = E[|A_i|] = n \prod_{j=0}^{i-1}q_j = n^{1-\frac{\nu}{3}\sum_{j=0}^{i-1}(4/3)^j}=n^{1-((4/3)^i-1)\nu}~.\]

For every $i \in [k-1]$, define the pivot $p_i(v)$ to be the closest vertex in $A_i$ to $v$, breaking ties by lexicographic order. For every $u \in A_i \backslash A_{i+1}$ define the {\em half bunch}
\begin{equation}\label{eq:hb}
B_{1/2}(u) = \{v \in A_i ~:~ d_G(u, v) < d_G(u, A_{i+1})/2\}~.
\end{equation}
Let $H = \{P_{uv} ~:~ u \in V, v \in B_{1/2}(u) \cup \{p_j(u) | i < j < k \}\}$, where $P_{uv}$ is the shortest path between $u,v$ in $G$ (if there is more than one, break ties consistently, by vertex id, say).

\begin{theorem}\label{1epsspan_theorem}
	For any weighted graph $G = (V, E)$ on $n$ vertices, and any integer $k > 2$, there exists $H$ of size at most $O(kn + n^{1+1/((4/3)^k - 1)})$, which is a $(1+\varepsilon, \beta\cdot W)$-spanner for any $0<\varepsilon <1$, with $\beta = O(k/\varepsilon)^{k-1}$.
\end{theorem}

\begin{theorem}
	\label{3span_theorem}
	For any weighted graph $G = (V, E)$ on $n$ vertices, and any integer $k > 2$, there exists $H$ of size at most $O(kn + n^{1+1/((4/3)^k - 1)})$, which is a $(3+\varepsilon, \beta\cdot W)$-spanner for any $\varepsilon>0$ with $\beta =O(1+1/\varepsilon) ^ {k-1}$.
\end{theorem}

Full proofs for both theorems are given in appendix \ref{appendix:spanner_proofs}.

\section{Efficient Implementation}\label{sec:apps}

Since we use very similar constructions to the ones in \cite{EN19hop}, we can use their efficient implementations (connecting to all pivots, which is the difference between constructions, can be done efficiently in their framework as well). We consider here the standard model of computation, and the PRAM (CRCW) model.
Given a parameter $1/k<\rho<1/2$, we will want poly-logarithmic parallel time and $\tilde{O}(|E|\cdot n^\rho)$ work / centralized time. This is achieved by adding additional $\lceil 1/\rho\rceil$ sets $A_i$, that are sampled with uniform probability $n^{-\rho}$, which in turn increases the exponent of $\beta$ by an additive $1/\rho+1$. (In the case of multiplicative stretch $1+\epsilon$, it also increases the base of the exponent in $\beta$.)

We summarize the efficient implementation results for hopsets and emulators in the following theorem.

\begin{theorem}
	\label{theorem-eff}
	For any weighted graph $G = (V, E)$ on $n$ vertices, parameters $k > 2$ and $1/k<\rho<1/2$, there is a randomized algorithm running in time $\tilde{O}(|E|\cdot n^\rho)$, that w.h.p. computes $H$ of size at most $O(kn + n^{1+1/(2^k - 1)})$, such that for any $0<\varepsilon<1$ this $H$ is:
	\begin{enumerate}
		\item A $(1+\varepsilon, \beta\cdot W)$-emulator with  $\beta=O\left(\frac{k+1/\rho}{\varepsilon}\right) ^ {k+1/\rho}$.
		\item A $(3+\varepsilon, \beta\cdot W)$-emulator with  $\beta=O(1/\varepsilon) ^ {k+1/\rho}$.
		\item A $(3+\varepsilon, \beta)$-hopset with $\beta = O(1/\varepsilon) ^ {k+1/\rho}$.

	\end{enumerate}
	
	Given $\varepsilon$ in advance, the algorithm can also be implemented in the PRAM (CRCW) model, in parallel time $\left(\frac{\log n}{\varepsilon}\right)^{O(k+1/\rho)}$ and work $\tilde{O}(|E|\cdot n^\rho)$, while increasing the size of $H$ by a factor of $O(\log^*n)$.
\end{theorem}

For spanners, recall that in Section \ref{sec:spanner} we have a somewhat different construction, and in the analysis we enforce a stricter requirement on the sampling probabilities $q_i$. To handle this, we start sampling with the uniform probability $n^{-\rho}$ only when $N_i\le n^{1-3\rho}$ (and not when $N_i\le n^{1-\rho}$ like before). Now the bound of Claim \ref{claim:span} still holds, as $N_i/q_i^3\le n$ even for these latter sets. The "price" we pay for waiting until $N_i\le n^{1-3\rho}$ is that the work will now be $|E|\cdot n^{3\rho}$. Rescaling $\rho$ by 3, we get the following.


\begin{theorem}
	\label{theorem-eff-span}
	For any weighted graph $G = (V, E)$ on $n$ vertices, parameters $k > 6$ and $1/k<\rho<1/6$, there is a randomized algorithm running in time $\tilde{O}(|E|\cdot n^\rho)$, that w.h.p. computes $H$ of size at most $O(kn + n^{1+1/(2^k - 1)})$, such that for any $0<\varepsilon<1$ this $H$ is:
	\begin{enumerate}
		\item A $(1+\varepsilon, \beta\cdot W)$-spanner with  $\beta=O\left(\frac{k+1/\rho}{\varepsilon}\right)^ {k+3/\rho}$.
		\item A $(3+\varepsilon, \beta\cdot W)$-spanner with  $\beta=O(1/\varepsilon) ^ {k+3/\rho}$.
	\end{enumerate}
	Given $\varepsilon$ in advance, the algorithm can also be implemented in the PRAM (CRCW) model, in parallel time $\left(\frac{\log n}{\varepsilon}\right)^{O(k+1/\rho)}$ and work $\tilde{O}(|E|\cdot n^\rho)$, while increasing the size of $H$ by a factor of $O(\log^*n)$.
\end{theorem}

\section{Almost Shortest Paths in Weighted Graphs}\label{sec:ASP}

Given a weighted graph $G=(V,E, w)$ with $n$ vertices and a set $S\subseteq V$ of $s$ sources, fix parameters $k>6$, $0<\varepsilon<1$ and $0<\rho<1/6$. Here we show that our hopsets, emulators and spanners can be used for a $(1+\varepsilon,\beta\cdot W)$-approximate shortest paths for pairs in $S\times V$, in various settings.

For the standard centralized setting, we first compute a $(1+\varepsilon, \beta\cdot W)$-spanner $H$ of size $O(kn + n^{1+1/(2^k - 1)})$ with $\beta=O(\frac{k+3/\rho}{\varepsilon}) ^ {k+3/\rho}$, in time $\tilde{O}(|E|\cdot n^\rho)$ as in Theorem \ref{theorem-eff-span}. Next, for every $u\in S$ run Dijkstra's shortest path algorithm in $H$, which takes time $O(s\cdot(|E(H)|+n\log n)) = \tilde{O}(s\cdot n^{1+1/(2^k-1)})$. 

The total running time for computing $(1+\varepsilon, \beta\cdot W)$-approximate shortest path for all $S\times V$, is $\tilde{O}(|E|\cdot n^\rho+s\cdot n^{1+1/(2^k-1)})$. One may choose $\rho=1/k$ and obtain the following.

\begin{theorem}\label{thm:paths}
	For any weighted graph $G = (V, E)$ on $n$ vertices, a set $S\subseteq V$ of $s$ sources, and parameters $k > 6$ and $0<\varepsilon<1$, there is a  randomized  algorithm running in time $\tilde{O}(|E|\cdot n^{1/k}+s\cdot n^{1+1/(2^k-1)})$, that computes $(1+\varepsilon,\beta\cdot W)$-approximate shortest paths for all pairs in $S\times V$, where
	$	\beta=\left(\frac{k}{\varepsilon}\right) ^ {O(k)}$.
\end{theorem}

We remark that there is a more general tradeoff by choosing $\rho$ as a free parameter. In addition, if one desires improved additive stretch, simply use the $(3+\varepsilon, \beta\cdot W)$-spanner with  $\beta=\varepsilon ^ {-O(k)}$.


\subsection{PRAM Shortest Paths and Distance Oracles}

Given a weighted graph $G=(V,E, w)$ with $n$ vertices, fix parameters $k\ge 1$, $0<\varepsilon<1$ and $0<\rho<1/6$. The first step in both settings (computing approximate shortest paths and distance oracles) is the same. We construct a $(1+\varepsilon, \beta\cdot W)$-emulator $G'$ of size $O(kn + n^{1+1/(2^k - 1)})\cdot\log^*n$ with  $\beta=O\left(\frac{k+1/\rho}{\varepsilon}\right) ^ {k+1/\rho}$, in parallel time $\left(\frac{\log n}{\varepsilon}\right)^{O(k+1/\rho)}$ and work $\tilde{O}(|E|\cdot n^\rho)$ as in \theoremref{theorem-eff}. Next, compute a $(1+\varepsilon, \beta)$-hopset $H$ size $O(n^{1+1/(2^k - 1)})\cdot\log^*n$ for $G'$ with the same $\beta = O\left(\frac{k+1/\rho}{\varepsilon}\right) ^ {k+1/\rho}$ within the same parallel time and work \cite{EN19hop}.\footnote{We remark that even though the emulator and hopset have exactly the same construction, we run the hopset algorithm on $G'$ and not on $G$, thus we get a different set of edges.} We store both $G'$ and $H$.

\paragraph{Approximate Shortest Paths.} 
For the sake of simplicity we will choose $\rho=1/k$.
Given a set $S$ of $s$ sources, for each source $u\in S$ run $\beta$ rounds of the Bellman-Ford algorithm in the graph $G'\cup H$ starting at $u$. In each round of Bellman-Ford, every vertex sends its neighbors the current distance estimate to $u$ that it has, and they update their distance estimate if needed. Since $G'\cup H$ is a sparse graph with $\tilde{O}(n^{1+1/(2^k - 1)})$ edges, with  $\tilde{O}(n^{1+1/(2^k - 1)})$ processors one can implement each iteration in PRAM (CRCW) in $O(2^k)$ time (see \cite{EN19hop} for more details). As we have only $\beta$ rounds, the total parallel time for all the Bellman-Ford rounds from all vertices in $S$ is $(k/\varepsilon)^{O(k)}$, and the total work is $\tilde{O}(s\cdot n^{1+1/(2^k - 1)})$.

As the error of the emulator is $(1+\varepsilon,\beta\cdot W)$, and the hopset has only multiplicative $1+\varepsilon$ stretch, the total error is only $(1+O(\varepsilon),O(\beta\cdot W)$. We thus have the following result.

\begin{theorem}\label{thm:pram-paths}
	For any weighted graph $G = (V, E)$ on $n$ vertices, a set $S\subseteq V$ of $s$ sources, and parameters $k > 2$ and $0<\varepsilon<1$, there is a PRAM randomized  algorithm running in $\left(\frac{\log n}{\varepsilon}\right)^{O(k)}$ parallel time and using $\tilde{O}(|E|\cdot n^{1/k}+s\cdot n^{1+1/(2^k - 1)})$ work, that computes $(1+\varepsilon,\beta\cdot W)$-approximate shortest paths for all pairs in $S\times V$, where
	$\beta=\left(\frac{k}{\varepsilon}\right) ^ {O(k)}$.
\end{theorem}

As above, we can get a more general tradeoff with the parameter $\rho$, and improve the additive stretch by using the hopsets and emulators from Sections \ref{sec:hopset},\ref{sec:3+eps-emul}, albeit the multiplicative stretch will increase.

\paragraph{Distance Oracles.}
Recall that we store the emulator $G'$ and a hopset $H$ for $G'$.
Whenever a query $u\in V$ arrives, we run $\beta$ rounds of Bellman-Ford algorithm in the graph $G'\cup H$. As noted above, each round of Bellman-Ford can be implemented in PRAM (CRCW) in $O(2^k)$ time using $\tilde{O}(n^{1+1/(2^k - 1)})$ processors. So the total parallel time for the query is $O\left(\frac{k+1/\rho}{\varepsilon}\right) ^ {k+1/\rho}$.  Rescaling $\varepsilon$, we get a $\left(1+\varepsilon, O\left(\frac{k+1/\rho}{\varepsilon}\right) ^ {k+1/\rho}\cdot W\right)$-approximation.
We conclude that the properties of the distance oracle we devise are:
\begin{itemize}
	\item Has size $O(kn + n^{1+1/(2^k - 1)})\cdot\log^*n$.
	\item Given query $u\in V$, can report $\left(1+\varepsilon, \beta\cdot W\right)$-approximation to {\em all} distances in $\{u\}\times V$, with $\beta=O\left(\frac{k+1/\rho}{\varepsilon}\right) ^ {k+1/\rho}$.
	\item Has query time $O\left(\frac{k+1/\rho}{\varepsilon}\right) ^ {k+1/\rho}$ and $\tilde{O}(n^{1+1/(2^k - 1)})$ work.
	\item The preprocessing time is $\left(\frac{\log n}{\varepsilon}\right)^{O(k+1/\rho)}$ and work $\tilde{O}(|E|\cdot n^\rho)$.
\end{itemize}


\bibliographystyle{alpha}
\bibliography{hopset}

\appendix

\section{A $(3+\epsilon,\beta)$-Hopset}
\label{sec:hopset}
Here we show that the set $H$ of Section \ref{sec:const} serves as a  $(3+\epsilon,\beta)$-hopset, for all $0<\epsilon<12$ simultaneously, with $\beta=2^{O(k\cdot\log(1/\epsilon))}$.

Denote by $d_G^{(t)}(u, v)$ the length of the shortest path between $u, v$ in $G$ that contains at most $t$ edges.
The following lemma bounds the number of hops and the stretch of the constructed hopset:

\begin{lemma}
	\label{3eps_hopset_lemma}
	Fix any $0 < \delta \leq 1/4$ and any $x,y \in V$. Then for every $0 \leq i \leq k-1$, at least one of the following holds:
	
	\begin{enumerate}
		\item $d_{G\cup H}^{(2 \cdot (1/\delta)^i - 1)}(x, y) \leq (3 + \frac{12 \delta}{1 - 3 \delta}) d_G(x, y)$
		\item $d_{G\cup H}^{(1)}(x, p_{i+1}(x)) \leq \frac{3}{1-3 \delta} d_G(x, y)$
	\end{enumerate}
\end{lemma}

\begin{proof}
	The proof is by induction on $i$. For the base case $i=0$, if $y \in B(x)$, then the edge $(x, y)$ was added to the hopset and the first item holds. If not, it means that $d_G(x,y) \geq d_G(x, p_1(x))$.
	Because every vertex is connected by a direct edge to all its pivots, the second item holds (since the coefficient of the right hand side is between 3 and 12 for $0<\delta\le 1/4$).
	
	Assume the claim holds for $i$, and we will prove it holds for $i+1$.  Partition the shortest path between $x$ and $y$ into $J \leq 1/\delta$ segments $\{L_j = [u_j, v_j]\}_{j\in[J]}$ each of length at most $\delta \cdot d_G(x, y)$, and at most $(1/\delta - 1)$ edges $\{(v_j , u_{j+1})\}_{j\in[J]}$ between consecutive segments. We can use the following: setting $u_1 = x$, and for each $j \in [J]$, set $v_j$ as the vertex in the shortest path between $u_j$ and $y$ which is farthest from $u_j$, but still $d_G(u_j, v_j) \leq \delta\cdot d_G(x, y)$.
	If $v_j \neq y$, set $u_{j+1}$ as the vertex which follows $v_j$ in the shortest path between $x$ and $y$. Otherwise set $u_{j+1} = y$. This partition satisfies our requirement $J \leq 1/\delta$ because for every $j \in [J-1]$, $d_G(u_j, u_{j+1}) > \delta\cdot d_G(x, y)$ (otherwise, we could have chosen $v_j$ as $u_{j+1}$).
	
	Next, apply the induction hypothesis for all the pairs $(u_j, v_j)$ with parameter $i$.	
	If for all the pairs $(u_j, v_j)$ the first item holds, we can show that the first item holds for $(x, y)$ with parameter $i+1$.
	Consider the path from $x$ to $y$ which uses the guaranteed path in $G\cup H$ of the first item for all the pairs $(u_j, v_j)$, and the edges $(v_j, u_{j+1})$. The number of hops in this path is bounded by $(1/\delta)\cdot(2(1/\delta)^i - 1) + (1/\delta - 1) \leq 2 \cdot (1/\delta)^{i+1} - 1$. The length of the path is bounded by (using the induction hypothesis on each pair):
	
	\begin{align*}
	d_{G\cup H}^{(2(1/\delta)^{i+1}-1)}(x, y) &\leq \sum_{j\in[J]}{(d_{G\cup H}^{(2(1/\delta)^i-1)}(u_j, v_j) + d^{(1)}_G(v_j, u_{j+1}))} \\
	&\leq  (3 + \frac{12 \delta}{1 - 3 \delta})d_G(x, y)~.
	\end{align*}
	
	Otherwise, there exist at least one segment for which the first item doesn't hold. Let $l \in [J]$ be the smallest index so that only the second item holds for the pair $(u_l, v_l)$. By the induction hypothesis
	
	\begin{align*}
	d_{G\cup H}^{(1)}(u_l, p_{i+1}(u_l)) \leq  \frac{3}{1-3 \delta} d_G(u_l, v_l) \leq   \frac{3\delta}{1-3 \delta} d_G(x, y).
	\end{align*}
	
	Since we added the edges $(x,p_{i+1}(x))$, $(y,p_{i+1}(y))$ to the hopset $H$, by the triangle inequality,
	\begin{align}
	d_{G \cup H}^{(1)}&(x, p_{i+1}(x)) \leq d_G(x, u_l) + d_{G}(u_l, p_{i+1}(u_l)) \nonumber\\
	 &\leq d_G(x, u_l) +  \frac{3\delta}{1-3 \delta}d_G(x, y), \label{eq:1}\\
	d_{G \cup H}^{(1)}&(y, p_{i+1}(y)) \leq d_G(y, u_l) + d_{G}(u_l, p_{i+1}(u_l)) \nonumber \\
	&\leq d_G(y, u_l) +  \frac{3\delta}{1-3 \delta}d_G(x, y). \label{eq:2}
	\end{align}
	
	If the edge $(p_{i+1}(x), p_{i+1}(y))$ exists in $H$, its length can be bounded by
	
	\begin{align}
	d_{G \cup H}^{(1)}&(p_{i+1}(x), p_{i+1}(y)) \nonumber \\
	&\leq d_{G\cup H}^{(1)}(p_{i+1}(x), x) + d_{G}(x, y) +  d_{G\cup H}^{(1)}(y, p_{i+1}(y)).\label{eq:3}
	\end{align}
	
	Thus, the distance between $x$ and $y$ using 3 hops is
	
	\begin{align*}
	d&_{G\cup H}^{(3)}(x, y) \\
	& \leq d_{G \cup H} ^{(1)} (x, p_{i+1}(x)) + d_{G \cup H} ^ {(1)}(p_{i+1}(x), p_{i+1}(y)) + d_{G \cup H}^{(1)}(p_{i+1}(y), y) \\
	& \stackrel{(\ref{eq:3})}{\leq} 2 d_{G \cup H} ^{(1)} (x, p_{i+1}(x)) + d_{G}(x, y) + 2 d_{G \cup H}^{(1)}(p_{i+1}(y), y) \\
	& \stackrel{(\ref{eq:1}) + (\ref{eq:2})}{\leq} 2(d_G(x, u_l) + \frac{3\delta}{1-3 \delta} d_G(x, y)) + d_G(x, y) + 2 (d_G(y, u_l) \\ 
	\frac{3\delta}{1-3 \delta} d_G(x, y)) \\
	& \leq  (3 + \frac{12 \delta}{1 - 3 \delta})d_G(x, y),
	\end{align*}
	
	therefore the first item holds.
	
	If $(p_{i+1}(x), p_{i+1}(y)) \notin H$, then we know that 
	$d_{G \cup H} ^{(1)}(p_{i+1}(x), p_{i+2}(p_{i+1}(x))) \leq d_G(p_{i+1}(x), p_{i+1}(y))$. We can bound the distance $d_{G \cup H}^{(1)}(x, p_{i+2}(x))$ using the triangle inequality:
	
	\begin{align*}
	d&_{G \cup H}^{(1)}(x, p_{i+2}(x)) \\
	&\leq d_G(x, p_{i+1}(x)) + d_G(p_{i+1}(x), p_{i+2}(p_{i+1}(x))) \\
	&\stackrel{(\ref{eq:3})}{\leq}  2 d_{G}(p_{i+1}(x), x) + d_{G}(x, y) +  d_{G}(y, p_{i+1}(y)) \\
	& \leq 2 (d_G(x, u_l) + \frac{3\delta}{1-3 \delta} d_G(x, y)) + d_G(x, y) + d_G(y, u_l) 
	+ \frac{3\delta}{1-3 \delta} d_G(x, y)\\
	& \leq 3 d_{G}(x, y) + \frac{9\delta}{1-3 \delta} d_G(x, y) =\frac{3}{1-3 \delta}d_G(x, y).
	\end{align*}
	
	Thus the second item holds.
	
\end{proof}
We conclude by summarizing the main result of this section.
\begin{theorem}
	\label{hopset_theorem}
	For any weighted graph $G = (V, E)$ on $n$ vertices, and any $k \geq 1$, there exists $H$ of size at most $O(kn + n^{1+1/(2^k - 1)})$, which is a $(3+\varepsilon, \beta)$-hopset for any $0 < \varepsilon \leq 12$, with $\beta = 2(3 + 12/\varepsilon) ^ {k-1}$.
\end{theorem}
\begin{proof}
	Let $x, y \in V$.
	Apply lemma \ref{3eps_hopset_lemma} for $x, y$ with $\delta = \frac{\varepsilon}{12 + 3\varepsilon}$ and $i=k-1$. Since $A_k = \emptyset$, the first item must hold:
	\[
	d_{G\cup H}^{(2(3 + 12/ \varepsilon)^{k-1})}(x, y) \leq (3 + \varepsilon) d_G(x, y).
	\]
\end{proof}

\begin{remark}
	Note that at its lowest, the hopbound is $O(4^k)$, achieved with stretch 15.
\end{remark}

\section{A $(3 + \varepsilon,\beta\cdot W)$-Emulator}\label{sec:3+eps-emul}

In this section we show that the same $H$ constructed in Section~\ref{sec:const} can also serve as a $(3 + \varepsilon,\beta\cdot W)$ emulator for weighted graphs, for all values of $0<\varepsilon<1$ simultaneously.

Let $G = (V, E)$ be a weighted graph 
and let $k \geq 1$ and $\Delta > 3$ be given parameters (think of $\Delta=3+O(1/\epsilon)$).
Fix a pair $x,y\in V$. Define $D_{-1}=0$ and for any integer $i\ge 0$, let $D_i = W(x,y) \cdot \sum_{j=0}^{i} \Delta ^ j$.
We can easily verify that
\begin{align}
\label{di_to_di1}
D_{i+1} = \Delta\cdot D_i  + W(x,y).
\end{align}


\begin{lemma}
	\label{3emu_lemma}
	Let $0 \leq i \leq k$ and let $x,y \in V$ such that $d_G(x,y) \leq D_{i}$ and $d_H(x, p_i(x)) \leq \frac{2 \Delta}{\Delta - 3} D_{i-1}$. Define $m = \max\{\Delta D_{i-1}, d_G(x, y)\}$. Then at least one of the following holds:
	\begin{enumerate}
		\item $d_H(x, y) \leq (3 + \frac{8}{\Delta - 3}) m$.
		\item $d_H(x, p_{i+1}(x)) \leq \frac{2 \Delta}{\Delta - 3} D_i$.
	\end{enumerate}
\end{lemma}

\begin{proof}
	The proof is by induction on $i$. For the base case $i=0$, if $y \in B(x)$, the first item holds. Otherwise $d_H(x, p_1(x)) \leq d_G(x, y) \leq W(x,y) = D_0$, thus the second item holds.
	
	Assume the claim holds for $i$ and prove for $i+1$. By the triangle inequality:
	\begin{align}
	\label{piy}
	d_H(y, p_{i+1}(y)) \leq d_G(y,x) + d_G(x, p_{i+1}(x)).
	\end{align}
	
	If $p_{i+1}(y) \in B(p_{i+1}(x))$, we have
	\begin{align}
	\label{between_pivots}
	d_H(p_{i+1}(x), p_{i+1}(y)) \leq d_G(p_{i+1}(x), x) + d_G(x, y) + d_G(y, p_{i+1}(y)).
	\end{align}
	
	Thus, the distance between $x$ and $y$ is
	\begin{align*}
	d_H&(x, y) \\
	& \leq d_H (x, p_{i+1}(x)) + d_H(p_{i+1}(x), p_{i+1}(y)) + d_H(p_{i+1}(y), y) \\
	& \stackrel{(\ref{between_pivots})}{\leq} 2 d_H(x, p_{i+1}(x)) + d_{G}(x, y) + 2 d_H(p_{i+1}(y), y) \\
	& \stackrel{(\ref{piy})}{\leq} 2 d_H(x, p_{i+1}(x)) + d_{G}(x, y) + 2(d_G(y,x) + d_G(x, p_{i+1}(x)))  \\
	& \leq 3 d_G(x, y) + \frac{4 \cdot 2 \Delta}{\Delta - 3} D_i\\
	&{\leq} \left(3 + \frac{8}{\Delta - 3}\right) m,
	\end{align*}
	therefore the first item holds.
	
	If $p_{i+1}(y) \notin B(p_{i+1}(x))$, then we know that
	
	\begin{equation}\label{eq:out-bunch}
	d_H(p_{i+1}(x), p_{i+2}(p_{i+1}(x))) \leq d_G(p_{i+1}(x), p_{i+1}(y))~.
	\end{equation}
	We can bound the distance $d_H(x, p_{i+2}(x))$ as follows.
	\begin{align*}
	d_H&(x, p_{i+2}(x)) \\
	& \leq d_G(x, p_{i+1}(x)) + d_G(p_{i+1}(x), p_{i+2}(p_{i+1}(x))) \\
	& \stackrel{(\ref{eq:out-bunch})}{\leq }2 d_{G}(p_{i+1}(x), x) + d_{G}(x, y) +  d_{G}(y, p_{i+1}(y)) \\
	& \stackrel{(\ref{piy})}{\leq} 2 d_H(x, p_{i+1}(x)) + d_{G}(x, y) + d_G(y,x) + d_G(x, p_{i+1}(x))  \\
	& \leq 2 d_G(x, y) + \frac{3 \cdot 2 \Delta}{\Delta - 3} D_i \\
	&\stackrel{(\ref{di_to_di1})}\leq 2 D_{i+1} + \frac{6}{\Delta - 3} D_{i+1} \\
	& = \frac{2 \Delta}{\Delta - 3} D_{i+1}.
	\end{align*}
	Hence the second item holds.
\end{proof}
We are now ready to state the result of this section.

\begin{theorem}
	\label{3emu_theorem}
	For any weighted graph $G = (V, E)$ on $n$ vertices, and any $k \geq 1$, there exists $H$ of size at most $O(kn + n^{1+1/(2^k - 1)})$, which is a $(3+\varepsilon, \beta\cdot W)$-emulator for any $\varepsilon>0$ with $\beta = O(1+1/\varepsilon) ^ {k-1}$.
\end{theorem}

\begin{proof}
	Let $x, y \in V$.	Recall that $P_{xy}$ is the shortest path between $x$ and $y$ in $G$, and fix $\Delta = 3 + 8/\varepsilon$.	
	Initialize $i = 0$. Let $z$ be the farthest vertex from $x$ in $P_{xy}$ satisfying $d_G(x,z) \leq D_i$. Note the requirement $d_H(x, p_0(x)) \leq \frac{2 \Delta}{\Delta - 3} D_{-1}=0$ holds since $p_0(x)=x$. Apply lemma \ref{3emu_lemma} on $x,z$ and $i$. If the second item holds, we increase $i$ by one, update $z$ to be the last vertex in $P(x, y)$ satisfying $d_G(x,z) \leq D_i$, and apply the lemma again for $x, z$ and $i$ (since the second item held for $i-1$, we have that $d_H(x, p_i(x)) \leq \frac{2 \Delta}{\Delta - 3} D_{i-1}$ indeed holds).
	
	Consider now the index $i$ such that the first item holds (we must find such an index, since at $i=k-1$ there is no pivot in level $k$). If it is the case that $d_G(x, y) \geq D_{i}$ then since $D_i - \Delta D_{i-1} = W(x,y)$, it must be that $d_G(x, z) \geq \Delta D_{i-1}$, as otherwise we could have taken a further away $z$ (recall that every edge on this path has weight at most $W(x,y)$). Therefore $m = d_G(x,z)$ and we found a path in $H$ from $x$ to $z$ with stretch at most $3 + \frac{8}{\Delta - 3}$. Next we update $x=z$, $i = 0$ and repeat the same procedure all over again.
	
	The last remaining case is that we found an index $i$ such that the first item holds but $d_G(x, y) < D_i$. Note that in such a case it must be that $z = y$. The path in $H$ we have from $x$ to $y$ is of length at most $(3+\frac{8}{\Delta - 3})\cdot D_i=(3+\varepsilon)\cdot D_i$. As $i\le k-1$ and $D_{k-1}\le 2\Delta^{k-1}\cdot W(x,y)$, we have that
	\[
	d_H(x,y)\le 2(3+\varepsilon)\cdot (3 + 8/\varepsilon) ^ {k-1}\cdot W(x,y)~,
	\]
	which is our additive stretch $\beta$.
	
\end{proof}

\section{Full proof of lemma \ref{lem:size_hopset_emu}}
\label{size_analysis_hop_emu}
If we order the vertices in $A_i$ by their distance to $u$, it is easy to see that the number of vertices which are in $A_i$ and closer than $p_{i+1}(u)$ is bounded by a random variable sampled from a geometric distribution with parameter $q_i$.
Hence $E[|B(u)|] \leq k + 1/q_i = k + n^{2^i\nu}2^{2^i + 1}$.
For $u \in A_{k-1}$, since $p_k(u)$ doesn't exist, $B(u)$ contains all the vertices in $A_{k-1}$.
The number of vertices in $A_{k-1}$ is a random variable sampled from binomial Distribution with parameters $(n, \prod_{j=0}^{k-2}q_j)=(n,n^{-(2^{k-1}-1)\nu} \cdot 2 ^ {-2^{k-1}-k+2})$.
Hence, the expected number of edges added by bunches of vertices in $A_{k-1}$ is
\begin{align*}
E\left[\binom{|A_{k-1}|}{2}\right] & \leq E[|A_{k-1}|^2] = E[|A_{k-1}|]^2  + Var(|A_{k-1}|) \\
& = n^2\prod_{j=0}^{k-2}q_j^2 + n(1 - \prod_{j=0}^{k-2}q_j)\prod_{j=0}^{k-2}q_j \\
& \leq n^{2 - 2(2^{k-1}-1)\nu} \cdot 2 ^ {2(-2^{k-1}-k+2)} + n\cdot 2^{-2^{k-1}-k+2} \\
& \leq (n^{1 + \nu} + n) 2 ^ {3-k}.
\end{align*}
Hence, the total expected number of edges in $H$ is:
\begin{align*}
&\sum_{i=0}^{k-2}(N_i \cdot n^{2^{i}\nu}\cdot 2 ^ {2^{i} +1}) + E[|A_{k-1}|^2]  + kn \\
&= \sum_{i=0}^{k-2}(n^{1+\nu}\cdot 2 ^ {-i+2}) + E[|A_{k-1}|^2]  + kn \\
&=O(kn + n^{1+\nu})
\end{align*}

\section{Proofs of theorems \ref{1epsspan_theorem},\ref{3span_theorem}}
\label{appendix:spanner_proofs}

\subsection{Size analysis}
Recall the construction of $H$ described in Section \ref{sec:spanner}. Define the bunch $B(u)$ as in (\ref{eq:bunch}). The following lemma will be useful to bound the size of the spanner $H$.
\begin{lemma}\label{lem:paths}
	Fix $0\le i\le k-1$. Let $u,v,x,y\in A_i$ be such that $v\in B_{1/2}(u)$ and $y\in B_{1/2}(x)$, and $P_{uv}\cap P_{xy}\neq\emptyset$, then all four points are in $B(u)$, or all four are in $B(x)$.
\end{lemma}
\begin{proof}
	Assume w.l.o.g. that $P_{uv}$ is not shorter than $P_{xy}$. Let $z\in V$ be a point in the intersection of the two shortest paths, then
	\begin{align*}
	d_G(u,x) &\le d_G(u,z)+d_G(z,x) \le d_G(u,v)+d_G(y,x) \\ 
	&\le 2 d_G(u,v)<d_G(u, A_{i+1})~,
	\end{align*} 
	
	so $x\in B(u)$. The calculation showing $y\in B(u)$ is essentially the same.
\end{proof}
Define the shortest path between two vertices consistently, s.t. each subpath is also a shortest path. Therefore two shortest paths can have at most one common subpath.

Fix $0\le i\le k-2$, and consider the graph $G_i$ containing all the shortest paths $P_{uv}$ with $u\in A_i$ and $v\in B_{1/2}(u)$. We claim that the number of edges in $G_i$ is at most $O(n+C_i)$, where $C_i$ is the number of pairwise intersections between these shortest paths. This is because vertices participating in at most 1 path have degree at most 2, and each intersection increases the degree of one vertex by at most 2 (recall that shortest paths can meet at most once).
\begin{claim}\label{claim:span}
	$\E[|C_i|]\le O(n^{1+\nu})$.
\end{claim}
\begin{proof}
	By Lemma \ref{lem:paths} each intersecting pair of paths $P_{uv}$ and $P_{xy}$ we have that all four points belong to the same bunch. Thus, each $u\in A_i$ can introduce at most $|B(u)|^3$ pairwise intersecting paths.
	Recall that $|B(u)|$ is a random variable distributed geometrically with parameter $q_i$, so
	\begin{align*}
	\E[|B(u)|^3]&=\sum_{j=1}^\infty j^3\cdot q_i\cdot(1-q_i)^{j-1} \\
	&\le q_i\cdot\sum_{j=1}^\infty (1-q_i)^{j-1}\cdot j(j+1)(j+2)\le\frac{6}{q_i^3}~.
	\end{align*}
	Thus the expected number of intersections at level $i$ is at most
	\begin{align*}
\E\left[\sum_{u\in A_i}|B(u)|^3\right] &\le O(N_i/q_i^3)=O(n^{1-((4/3)^i-1)\nu}\cdot (n^{4^i\nu/3^{i+1}})^3) \\
&=O(n^{1+\nu})~.
	\end{align*}

	It remains to bound path intersections in the last level $k-1$. Recall that
	\[
	N_{k-1}= n^{1-((4/3)^{k-1}-1)\nu}=n^{1-((4/3)^{k-1}-1)/((4/3)^k-1)}\le n^{(1 + \nu)/4}~.
	\]
	Since the random choices for each point are independent, we have by Chernoff bound that $\Pr[N_{k-1}>2n^{(1 + \nu)/4}]\le e^{-\Omega(n^{(1 + \nu)/4})}$, so with very high probability the last set $A_{k-1}$ contains $O(n^{(1 + \nu)/4})$ points. It means that we have $O(\sqrt{n^{1 + \nu}})$ paths connecting these points, and even if they all intersect, they can yield at most $O(n ^ {1 + \nu})$ intersections.
\end{proof}

It remains to bound the number of edges to pivots. For every $v \in V$, $p_i(v)$ is the closest vertex to $v$ in $A_i$ breaking ties by id. Therefore all the vertices in $P_{v,p_i(v)}$ share the same pivot at level $i$. Thus we add at most one edge for each vertex at every level, and $O(nk)$ edges overall.

We conclude that the size of $H$ is at most $O(k\cdot n^{1+\nu})$ (we can slightly change the probabilities by introducing a factor of $2^{-4^{i}/3^{i+1} - 1}$ to obtain size $O(kn+ n^{1+\nu})$, as we did before.

\subsection{Proof of theorem \ref{1epsspan_theorem}}

We use the corresponding analysis of the emulator from Section \ref{sec:emulator}. The use of half-bunches instead of bunches creates the following version of Lemma \ref{1epsemu_constant_lemma}.

\begin{lemma}\label{1epsspan_constant_lemma}
	Fix $\Delta >5$. Let $0 \leq i < k$ and let $x,y \in V$ such that $d_G(x,y) \geq (3 \Delta) ^ iW(x,y)$. Then at least one of the following holds:
	\begin{enumerate}
		\item $d_H(x, y) \leq (1 + \frac{8i}{\Delta - 5}) d_G(x, y)$
		\item $d_H(x, p_{i+1}(x)) \leq \frac{2\Delta}{\Delta - 5} d_G(x, y)$
	\end{enumerate}
	
\end{lemma}

The main difference in the proof is in (\ref{eq:no-pivot-2}), which is replaced by
\[
d_H(p_{i+1}(u_l), p_{i+2}(p_{i+1}(u_l))) \leq 2d_G(p_{i+1}(u_l), p_{i+1}(u_r))~.
\]
The new bounds in the Lemma guarantee the calculations still go through. For Lemma \ref{1epsemu_beta_lemma} which takes care of small distances, we have the following change, with a very similar proof.

\begin{lemma}\label{1epsspan_beta_lemma}
	Let $0 \leq i < k$ and fix $x,y \in V$. Let 
	
	$ m =\max\{d_G(x, p_i(x)),d_G(y, p_i(y)),d_G(x,y)\}$. Then at least one of the following holds:
	\begin{enumerate}
		\item $d_H(x,y) \leq 5m$
		\item $d_H(x, p_{i+1}(x)) \leq 7m$
	\end{enumerate}
\end{lemma}

In the proof we set $\Delta=5+\frac{8(k-1)}{\varepsilon}$. The rest of the calculations follow analogously, one change is that when iteratively applying Lemma \ref{1epsspan_beta_lemma}, the bound $m$ increases by a factor of 7 (rather than 4, as in Lemma \ref{1epsemu_beta_lemma}), but as $7\le 3\Delta$ is still true, it does not change anything.

\subsection{Proof of theorem \ref{3span_theorem}}
The stretch analysis is very similar to that of the emulator from Section \ref{sec:3+eps-emul}, the main difference is in the use of half-bunches rather than the full ones, but this will increase the distance to pivots by a factor of 2, and affect the additive stretch only. 
We follow the analysis and notation presented in Section \ref{sec:3+eps-emul}, but with $\Delta>5$. We replace Lemma \ref{3emu_lemma} with the following.
\begin{lemma}
	\label{3span_lemma}
	Let $0 \leq i \leq k$ and let $x,y \in V$ such that $d_G(x,y) \leq D_{i}$ and $d_H(x, p_i(x)) \leq \frac{3\Delta}{\Delta - 5} D_{i-1}$. Define $m = \max\{\Delta D_{i-1}, d_G(x, y)\}$. Then at least one of the following holds:
	\begin{enumerate}
		\item $d_H(x, y) \leq (3 + \frac{16}{\Delta - 5}) m$.
		\item $d_H(x, p_{i+1}(x)) \leq \frac{4 \Delta}{\Delta - 5} D_i$.
	\end{enumerate}
\end{lemma}

The main difference in the proof is in (\ref{eq:out-bunch}), which is replaced by
\[
d_H(p_{i+1}(x), p_{i+2}(p_{i+1}(x))) \leq 2d_G(p_{i+1}(x), p_{i+1}(y))~,
\]
since we use half-bunches. One can then follow the calculations in the proof of Lemma \ref{3emu_lemma}, and check that the altered constants used in the 2 cases above suffice.

The proof of Theorem \ref{3span_theorem} is the same as the proof of Theorem \ref{3emu_theorem}, the only differences are taking $\Delta=5+\frac{16}{\varepsilon}$ (which affects the value of $\beta$) and using the bounds of Lemma \ref{3span_lemma} rather than of Lemma \ref{3emu_lemma}. 
\end{document}